\documentclass[11pt]{article}%
\usepackage{amsfonts}
\usepackage{amsmath}
\usepackage{amssymb}
\usepackage{graphicx}%
\providecommand{\U}[1]{\protect\rule{.1in}{.1in}}
\newtheorem{theorem}{Theorem}
\newtheorem{acknowledgement}[theorem]{Acknowledgement}

\newtheorem{notation}[theorem]{Notation}

\newtheorem{proposition}[theorem]{Proposition}

\newenvironment{proof}[1][Proof]{\noindent\textbf{#1.} }{\ \rule{0.5em}{0.5em}}
\begin{document}

\title{On the Use of Minimum Volume Ellipsoids and Symplectic Capacities \ for
Studying Classical Uncertainties for Joint Position--Momentum Measurements}
\author{Maurice A. de Gosson\thanks{E-mail: maurice.de.gosson@univie.ac.at}\\\ \textit{Universit\"{a}t Wien, NuHAG}\\\textit{Fakult\"{a}t f\"{u}r Mathematik }\\\textit{A-1090 Wien}}
\maketitle

\begin{abstract}
We study the minimum volume ellipsoid estimator associates to a cloud of
points in phase space. Using as a natural measure of uncertainty the
symplectic capacity of the covariance ellipsoid we find that classical
uncertainties obey relations similar to those found in non-standard quantum mechanics.

\end{abstract}

\textbf{Keywords}: position-momentum measurements, minimum volume ellipsoid,
symplectic capacity, uncertainty principle, Hamiltonian system

\section{Introduction}

Contrarily to what is often believed the Heisenberg uncertainty principle%
\begin{equation}
\Delta p_{j}\Delta x_{j}\geq\tfrac{1}{2}\hbar\text{ }\label{hup}%
\end{equation}
is not a statement about the accuracy of our measurement instruments; its
derivation assumes on the contrary perfect instruments. The correct
interpretation of Heisenberg's inequalities is the following (see e.g. Peres
\cite{Peres}, p.93): if the same preparation procedure is repeated a large
number of times, and is followed by either by a measurement of $x_{j}$, or by
a measurement of $p_{j}$, the results obtained have standard deviations
$\Delta x_{j}$ and $\Delta p_{j}$ satisfying (\ref{hup}). The same
interpretation is of course true for the stronger Robertson--Schr\"{o}dinger
\cite{Rob,Schr} inequalities%
\begin{equation}
(\Delta p_{j})^{2}(\Delta x_{j})^{2}\geq\Delta(x_{j},p_{j})^{2}+\tfrac{1}%
{4}\hbar^{2}\text{ }\label{robup}%
\end{equation}
to which (\ref{hup}) reduce if one neglects the covariances $\Delta
(x_{j},p_{j})^{2}$; they are complemented by the trivial inequalities%
\begin{equation}
\Delta p_{j}\Delta x_{k}\geq0\text{ if }j\neq k\text{, \ }\Delta p_{j}\Delta
p_{k}\geq0\text{ \ , \ }\Delta x_{j}\Delta x_{k}\geq0\label{noconj1}%
\end{equation}
which might be violated in a nonstandard form of quantum mechanics,
noncommutative mechanics (NCQM, see Dias et al. \cite{DGLP} and the references
therein), where the second and third inequalities (\ref{noconj1}) are be
replaced with%
\begin{equation}
\Delta p_{j}^{2}\Delta p_{k}^{2}\geq\Delta(p_{j},p_{k})^{2}+\tfrac{1}{4}%
\theta_{jk}^{2}\text{ , }\Delta x_{j}^{2}\Delta x_{k}^{2}\geq\Delta
(x_{j},x_{k})^{2}+\tfrac{1}{4}\eta_{jk}^{2}\label{noconj2}%
\end{equation}
where $\theta_{jk}^{2}=\theta_{kj}^{2}$ and $\eta_{jk}^{2}=\eta_{kj}^{2}$.

In \emph{classical} statistical mechanics the situation is somewhat different:
due to the inherent inaccuracy of the measurement apparatus there are
uncertainties for all pairs of variables, conjugate or not. We are going to
show in this article that despite the different nature of quantum and
classical uncertainties, the latter can be described in a similar way, leading
to inequalities which are formally identical to those of NCQM, namely%
\begin{align}
\Delta x_{j}^{2}\Delta x_{k}^{2}  &  \geq\Delta(x_{j},x_{k})^{2}+a_{jk}^{2}\\
\Delta p_{j}^{2}\Delta p_{k}^{2}  &  \geq\Delta(p_{j},p_{k})^{2}+c_{jk}^{2}\\
\Delta x_{j}^{2}\Delta p_{k}^{2}  &  \geq\Delta(x_{j},p_{k})^{2}+b_{jk}^{2}%
\end{align}
with $a_{jk}^{2}=a_{kj}^{2}$, $b_{jk}^{2}=b_{kj}^{2}$, $c_{jk}^{2}=c_{kj}^{2}%
$. We have actually already shown in a recent paper \cite{FP} (also see de
Gosson and Luef \cite{physreps}) that inequalities of the type (\ref{hup}%
)--(\ref{robup}) are by no means characteristic of quantum mechanics, and that
there are formally similar statements in classical statistical mechanics; to
sustain our claim we used tools from robust multivariate statistics (the
Minimum Volume Ellipsoid method, reviewed in Section \ref{sec2}) together with
a topological device, the notion of symplectic capacity, which we use as a
natural device for measuring uncertainty. The symplectic capacity of a closed
phase space set is the equatorial area of the largest ball that can be sent in
this set using only symplectomorphisms (=canonical transformations). The
symplectic capacity of the covariance ellipsoid is not related to its volume
(it always has the dimension of an area), but rather to that of a classical
action. We note that symplectic capacities have been used by Scheeres and his
collaborators \cite{sch2,sch1} to study satellite guidance and constraints on
spacecraft trajectories.

The aim of this paper is to extend the results in \cite{FP} and to show that
these NCQM uncertainties also appear quite naturally in classical statistical mechanics.

\begin{notation}
We identify $\mathbb{R}^{n}\times\mathbb{R}^{n}$ with $\mathbb{R}^{2n}$ and
write $(x,p)=z$ if $x\in\mathbb{R}^{n}$ and $p\in\mathbb{R}^{n}$. We will view
$x,p$ and $z$ as column vectors in all matrix computations. A symplectic form
on $\mathbb{R}^{2n}$ is a non-degenerate bilinear antisymmetric form on that
space. Let $\alpha$ be a symplectic form on $\mathbb{R}^{2n}$; the
corresponding symplectic group is denoted $\operatorname*{Sp}(2n,\alpha)$ [it
is the group of all linear automorphisms $S$ of $\mathbb{R}^{2n}$ such that
$\alpha(Sz,Sz^{\prime})=\alpha(z,z^{\prime})$ for all vectors $z,z^{\prime}$
in $\mathbb{R}^{2n}$]. We denote by $\sigma$ the standard symplectic form on
$\mathbb{R}^{2n}$: $\sigma(z,z^{\prime})=(z^{\prime})^{T}Jz$ if $z=(x,p)$,
$z^{\prime}=(x^{\prime},p^{\prime})$; $J$ is the standard symplectic matrix $%
\begin{pmatrix}
0_{n\times n} & I_{n\times n}\\
-I_{n\times n} & 0_{n\times n}%
\end{pmatrix}
$.
\end{notation}

\section{\label{sec2}The Minimum Volume Ellipsoid Method}

Let us consider a system of point-like particles; we assume that there are $n$
degrees of freedom and label the position coordinates and momenta
$x=(x_{1},...,x_{n})$, $p=(p_{1},...,p_{n})$; we will also use the collective
notation $z=(x,p)$. Assume now that we perform position and momentum
measurements on a large number $K$ of identical copies of that system; we get
a cloud $\mathcal{S}=\{z_{1},z_{2},...,z_{N}\}$, $N=nK$, of points in
$\mathbb{R}^{2n}$. An efficient method for studying that cloud consists in
using the minimum volume ellipsoid (MVE) method for the multivariate location
and scatter (Rousseeuw \cite{Rou}). Geometrically speaking this method is an
application of the John--L\"{o}wner theorem (see for instance Ball
\cite{Ball}): consider a subset $\{z_{i_{1}},z_{i_{2}},...,z_{i_{k}}\}$ of
$\mathcal{S}$; we will assume that the points $z_{i_{j}}$ are in general
position, i.e. that they do not remain in some hyperplane of $\mathbb{R}^{2n}%
$. This condition is sufficient and necessary for any ellipsoid containing
these points to have positive volume. The points $z_{i_{1}},z_{i_{2}%
},...,z_{i_{k}}$ determine a polyhedron in $\mathbb{R}^{2n}$; we denote by
$\mathcal{K}_{k}$ the convex hull of that polyhedron (it is the smallest
subset of $\mathbb{R}^{2n}$ containing $\{z_{i_{1}},z_{i_{2}},...,z_{i_{k}}%
\}$). The John--L\"{o}wner theorem ensures us that there exists a unique
ellipsoid $\mathcal{J}_{k}$ in $\mathbb{R}^{2n}$ containing $\mathcal{K}_{k}$
and having minimum volume among all the ellipsoids having this property.
Repeating this process for all subsets of the cloud $\mathcal{S}$ having $k$
elements in general position we get a family of ellipsoids; by definition the
MVE is the one with the smallest volume.

This method is practically implemented as follows (see e.g. Van Aelst and
Rousseeuw \cite{aero}, Hubert et al. \cite{hurovan}): choose an integer $k$
between $[N/2]+1$ and $N$ ($[N/2]$ the integer part of $N/2$); this constant
$k$ determines the robustness of the resulting estimator; a common choice is
\begin{equation}
k=\left[  \frac{N+2n+1}{2}\right]  . \label{k}%
\end{equation}
By definition, the location and scatter estimators minimize the determinant of
the matrices $M$ subject to the condition
\begin{equation}
\#\left\{  j:(z_{j}-\bar{z})^{T}M^{-1}(z_{j}-\bar{z})\leq m^{2}\right\}  \geq
k \label{jl2}%
\end{equation}
where the minimization is over all $\bar{z}\in\mathbb{R}^{2n}$ and all
positive definite symmetric matrices $M$ of size $2n$. Here $m$ is a fixed
constant, chosen so that the MVE estimator is a consistent estimator for of
the covariance matrix for data coming from a multivariate normal distribution,
that is
\[
m=\sqrt{\chi_{2n,\alpha}^{2}}\text{ \ , \ }\alpha=k/N
\]
where $\chi_{2n,\alpha}^{2}$ is a chi-square distribution with $2n$ degrees of
freedom (see Lopuh\"{a}a and Rousseeuw \cite{loro}). One the pair $(M,\bar
{z})$ is determined, the minimum volume ellipsoid (MVE) is the set of all $z$
in $\mathbb{R}^{2n}$ such that
\begin{equation}
(z-\bar{z})^{T}M^{-1}(z-\bar{z})\leq m^{2}. \label{jl1}%
\end{equation}

The next step consists in associating to the MVE $\mathcal{J}$ a covariance
matrix. For this one has to choose an adequate value $m_{0}$ for $m$; denoting
the corresponding matrix $M$ by $\Sigma$ the MVE is the ellipsoid%
\begin{equation}
\mathcal{C}:(z-\bar{z})^{T}\Sigma^{-1}(z-\bar{z})\leq m_{0}^{2} \label{cov1}%
\end{equation}
and $\Sigma$ is then precisely the covariance matrix.

We will write the covariance matrix in the form%
\begin{equation}
\Sigma=%
\begin{pmatrix}
\Delta(x,x) & \Delta(x,p)\\
\Delta(p,x) & \Delta(p,p)
\end{pmatrix}
\label{sigmablock1}%
\end{equation}
where $\Delta(x,x)=\left(  (\Delta(x_{i},x_{j})\right)  _{1\leq i,j\leq n}$
and so on. We will write, as is customary, $\Delta(x_{i},x_{i})=$ $\Delta
x_{i}^{2}$ and $\Delta(p_{i},p_{i})=$ $\Delta p_{i}^{2}$; we have
$\Delta(x_{i},x_{j})=\Delta(x_{j},x_{i})$, $\Delta(p_{i},p_{j})=\Delta
(p_{j},p_{i})$, $\Delta(x_{i},p_{j})=\Delta(p_{j},x_{i})$ hence the covariance
matrix is symmetric.

\section{A Condition on $\Sigma$}

Let $A=(a_{jk})_{1\leq j,k\leq n}$ and $C=(c_{jk})_{1\leq j,k\leq n}$ be two
real antisymmetric matrices, and $B=(b_{jk})_{1\leq j,k\leq n}$ a real
symmetric matrix. To $A,B,C$ we associate the $2n\times2n$ antisymmetric
matrix
\[
\Omega=%
\begin{pmatrix}
A & B\\
-B & C
\end{pmatrix}
\]
which is the most general form an antisymmetric of size $2n$ can have. We will
always assume that the matrix $\Omega$ is in addition invertible. This
condition implies that the bilinear form $\omega$ on $\mathbb{R}^{2n}$ defined
by $\omega(z,z^{\prime})=-(z^{\prime})^{T}\Omega^{-1}z$ is a symplectic form
on $\mathbb{R}^{2n}$, and we denote by $(\mathbb{R}^{2n},\omega)$ the
corresponding symplectic phase space. Notice that when $A=C=0$ and $B=I$ we
have $\Omega=J=%
\begin{pmatrix}
0 & I\\
-I & 0
\end{pmatrix}
$, the standard symplectic matrix. The general case can actually be reduced to
the standard case:

\begin{proposition}
\label{propf}There exist linear automorphisms $F$ of $\mathbb{R}^{2n}$ such
that $\Omega=F^{T}JF$; equivalently $F$ is a linear symplectomorphism
$(\mathbb{R}^{2n},\sigma)\longrightarrow(\mathbb{R}^{2n},\omega)$, that is we
have $\omega(Fz,Fz^{\prime})=\sigma(z,z^{\prime})$ for all $z$ and $z^{\prime
}$.
\end{proposition}

\begin{proof}
Let $\mathcal{B}^{\sigma}=\{e_{1}^{\sigma},...,e_{n}^{\sigma}\}\cup
\{f_{1}^{\sigma},...,f_{n}^{\sigma}\}$ and $\mathcal{B}^{\omega}%
=\{e_{1}^{\omega},...,e_{n}^{\omega}\}\cup\{f_{1}^{\omega},...,f_{n}^{\omega
}\}$ be symplectic bases of $(\mathbb{R}^{2n},\sigma)$ and $(\mathbb{R}%
^{2n},\omega)$ respectively (i.e. $\omega(e_{j}^{\omega},e_{k}^{\omega
})=\omega(f_{j}^{\omega},f_{k}^{\omega})=0$, $\omega(f_{j}^{\omega}%
,e_{k}^{\omega})=\delta_{jk}$ and similar relations for the $\sigma
(e_{j}^{\sigma},e_{k}^{\sigma})$, etc.). The automorphism $F$ defined by
$F(e_{j}^{\sigma})=e_{j}^{\omega}$, $F(f_{j}^{\sigma})=f_{j}^{\omega}$ is a
symplectomorphism $(\mathbb{R}^{2n},\sigma)\longrightarrow(\mathbb{R}%
^{2n},\omega)$.
\end{proof}

We remark that this result (which can also be proved using the properties of
the Pfaffian $Pf(\Omega)$) is just a restatement of the linear (and global)
version of Darboux's theorem \cite{HZ} on the local equivalence of all
symplectic manifolds with same dimension.

Let now $\Sigma$ be the covariance matrix as defined above, and consider the
matrix
\begin{equation}
\Sigma+i\Omega=%
\begin{pmatrix}
\Delta(x,x)+iA & \Delta(x,p)+iB\\
\Delta(p,x)-iB & \Delta(p,p)+iC
\end{pmatrix}
. \label{sigom}%
\end{equation}
We observe that $\Sigma+i\Omega$ is Hermitian since $\Sigma$ is symmetric and
$(i\Omega)^{\ast}=i\Omega$ since $\Omega$ is real antisymmetric. The
eigenvalues of $\Sigma+i\Omega$ \ are thus real. From now on we will assume
that these eigenvalue are nonnegative, that is $\Sigma+i\Omega$ is
semi-definite positive, which we write
\begin{equation}
\Sigma+i\Omega\geq0. \label{omega}%
\end{equation}
(One can actually show that this condition automatically implies that
$\Sigma>0$). Notice that we have in particular%
\begin{equation}
\Delta(x,x)+iA\geq0\text{ \ , \ }\Delta(p,p)+iC\geq0. \label{delta}%
\end{equation}

When $n=1$ the covariance matrix is just%
\[
\Sigma=%
\begin{pmatrix}
\Delta x^{2} & \Delta(x,p)\\
\Delta(p,x) & \Delta p^{2}%
\end{pmatrix}
\]
and the antisymmetric matrices $\Theta$ and $N$ are zero so that $\Omega=aJ=%
\begin{pmatrix}
0 & a\\
-a & 0
\end{pmatrix}
$. The condition $\Sigma+i\Omega=\Sigma+iaJ\geq0$ is in this case equivalent
to%
\[%
\begin{pmatrix}
\Delta x^{2}+ia & \Delta(x,p)\\
\Delta(p,x) & \Delta p^{2}-ia
\end{pmatrix}
\geq0
\]
which is in turn equivalent to the single inequality
\[
\Delta x^{2}\Delta p^{2}\geq\Delta(x,p)^{2}+a^{2}.
\]
This is of course formally the Robertson--Schr\"{o}dinger inequality
(\ref{robup}); in particular we have the Heisenberg-type inequality $\Delta
x\Delta p\geq a$.

Let us extend the study to higher dimensions. When $n=2$ the matrices $A,B,$
and $C$ are of the type%
\[
A=%
\begin{pmatrix}
0 & a\\
-a & 0
\end{pmatrix}
\text{, }B=%
\begin{pmatrix}
b & d\\
d & e
\end{pmatrix}
\text{, }C=%
\begin{pmatrix}
0 & c\\
-c & 0
\end{pmatrix}
\]
so that $\Sigma+i\Omega$ is the $4\times4$ matrix%
\[%
\begin{pmatrix}
\Delta x_{1}^{2} & \Delta(x_{1},x_{2})+ia & \Delta(x_{1},p_{1})+ib &
\Delta(x_{1},p_{2})+id\\
\Delta(x_{2},x_{1})-ia & \Delta x_{2}^{2} & \Delta(x_{2},p_{1})+id &
\Delta(x_{2},p_{2})+ie\\
\Delta(p_{1},x_{1})-ib & \Delta(p_{1},x_{2})-id & \Delta p_{1}^{2} &
\Delta(p_{1},p_{2})+ic\\
\Delta(p_{2},x_{1})-id & \Delta(p_{2},x_{2})-ie & \Delta(p_{2},p_{1})-ic &
\Delta p_{2}^{2}%
\end{pmatrix}
.
\]
Recalling Sylvester's criterion \cite{horn,johnson} which says that a
Hermitian matrix is positive semidefinite if an only if all of its principal
minors are nonnegative, the condition $\Sigma+i\Omega\geq0$ implies that the
principal minors of order two of $\Sigma+i\Omega$ must be $\geq0$, we
immediately get the inequalities
\begin{subequations}
\label{up}%
\begin{align*}
\Delta x_{1}^{2}\Delta x_{2}^{2}  &  \geq\Delta(x_{1},x_{2})^{2}+a^{2}\\
\Delta x_{1}^{2}\Delta p_{1}^{2}  &  \geq\Delta(x_{1},p_{1})^{2}+b^{2}\\
\Delta p_{1}^{2}\Delta p_{2}^{2}  &  \geq\Delta(p_{1},p_{2})^{2}+c^{2}\\
\Delta x_{1}^{2}\Delta p_{2}^{2}  &  \geq\Delta(x_{1},p_{2})^{2}+d^{2}\\
\Delta x_{2}^{2}\Delta p_{1}^{2}  &  \geq\Delta(x_{2},p_{1})^{2}+d^{2}\\
\Delta x_{2}^{2}\Delta p_{2}^{2}  &  \geq\Delta(x_{2},p_{2})^{2}+e^{2}.
\end{align*}

The same argument shows that more generally:
\end{subequations}
\begin{proposition}
Let $n\geq2$. If the covariance matrix $\Sigma$ satisfies the condition
$\Sigma+i\Omega\geq0$ then the following uncertainty relations hold:%
\begin{align}
\Delta x_{j}^{2}\Delta x_{k}^{2}  &  \geq\Delta(x_{j},x_{k})^{2}+a_{jk}%
^{2}\label{upc1}\\
\Delta p_{j}^{2}\Delta p_{k}^{2}  &  \geq\Delta(p_{j},p_{k})^{2}+c_{jk}%
^{2}\label{upc2}\\
\Delta x_{j}^{2}\Delta p_{k}^{2}  &  \geq\Delta(x_{j},p_{k})^{2}+b_{jk}^{2}.
\label{upc3}%
\end{align}
In particular, if $\Omega=\varepsilon J$, $\varepsilon>0$, these conditions
reduce to the Robertson--Schr\"{o}dinger inequalities
\[
\Delta x_{j}^{2}\Delta p_{k}^{2}\geq\Delta(x_{j},p_{k})^{2}+\varepsilon^{2}.
\]

\end{proposition}

A warning: the group of inequalities (\ref{upc1})--(\ref{upc3}) is \emph{not
equivalent} to the condition $\Sigma+i\Omega\geq0$ as soon as $n>1$. Here is a
simple counterexample in the case $n=2$ and $\Omega=J$: choose $\eta=1$ and%
\[
\Sigma=%
\begin{pmatrix}
1 & -1 & 0 & 0\\
-1 & 1 & 0 & 0\\
0 & 0 & 1 & 0\\
0 & 0 & 0 & 1
\end{pmatrix}
.
\]
This matrix is positive definite, and the inequalities above hold trivially
(they reduce to equalities); the matrix $\Sigma+iJ$ is however indefinite (its
determinant is $-1$).

\section{The Symplectic Capacity of an Ellipsoid}

For the basics of symplectic geometry we are going to use we refer to Arnol'd
\cite{Arnold}, de Gosson \cite{Birk} (Arnol'd uses the term \textquotedblleft
canonical transformation\textquotedblright\ for symplectomorphism; this
terminology is usual in Physics).

The condition $\Sigma+i\Omega\geq0$ can be restated in terms of the
\emph{symplectic capacity }of the minimum volume ellipsoid $\mathcal{C}_{0}$
given by (\ref{cov1}). Let us recall the definition of the notion of
symplectic capacity (of which we have given a detailed discussion in de Gosson
and Luef \cite{physreps}). Let us call symplectic manifold a submanifold $U$
of $\mathbb{R}^{2n}$ (possibly with boundary) equipped with a symplectic form
$\alpha$. A symplectic capacity associates to every symplectic manifold
$(U,\alpha)$ a number $c(U,\alpha)\geq0$ or $+\infty$; this correspondence
must satisfy the following axioms \cite{HZ,Polter}:

\begin{enumerate}
\item \textbf{Monotonicity}: If $\Phi:(U,\alpha)\longrightarrow(V,\beta)$ is a
symplectic embedding, i.e. a diffeomorphism satisfying $\beta(\Phi
(z),\Phi(z^{\prime})=\alpha(z,z^{\prime})$ we must have $c(U,\alpha)\leq
c(V,\beta)$;

\item \textbf{Conformality}: For every real $\lambda\neq0$ we have
$c(U,\lambda\alpha)=\lambda^{2}c(U,\alpha)$;

\item \textbf{Normalization}: $c(B(R),\sigma)=c(Z_{j}(R),\sigma)=\pi R^{2}$;
here $\sigma$ is the standard symplectic form, $B(R)$ the ball $|x|^{2}%
+|p|^{2}\leq R^{2}$ and $Z_{j}(R)$ the cylinder $x_{j}^{2}+p_{j}^{2}\leq
R^{2}$.
\end{enumerate}

In general symplectic capacities are not related to the notion of volume; this
is already clear from the normalization condition $c(Z_{j}(R),\sigma)=\pi
R^{2}$ which shows that as soon as $n>1$ the symplectic capacity of a region
with infinite volume can be finite; in fact the conformality axiom shows that
symplectic capacities behave as areas under dilations. Also note that
symplectic capacities are extensive quantities, i.e. they do not depend
directly on dimension, as volume does.

The monotonicity axiom implies that if there exists a symplectomorphism
(=symplectic diffeomorphism) $\Phi:(U,\alpha)\longrightarrow(V,\beta)$ such
that $\Phi(U)=V$ then%
\begin{equation}
c(U,\alpha)=c(\Phi(U),\beta)=c(V,\beta) \label{equal}%
\end{equation}
hence \emph{symplectic capacities are symplectic invariants}. The basic
example of a symplectic capacity is Gromov's width\footnote{Sometimes also
called symplectic area.} $c_{\mathrm{GR}}$. It is defined as follows: let
$R_{\mathrm{GR}}$ be the \textquotedblleft symplectic radius\textquotedblright%
\ of $U$: it is the supremum of all radii of balls $B(R)$ that can be embedded
in $U$ using symplectomorphisms of $(\mathbb{R}^{2n},\sigma)$. By definition
the Gromov width of $U$ is then $c_{\mathrm{GR}}(U,\sigma)=\pi R_{\mathrm{GR}%
}^{2}$ (with $c_{\mathrm{GR}}(U,\sigma)=\infty$ if $R_{\mathrm{GR}}=\infty)$.
The fact that the normalization axiom is satisfied follows from Gromov's
non-squeezing theorem \cite{Gromov}, a deep property of symplectic topology
which says that the ball $B(R)$ cannot be embedded inside a cylinder
$Z_{j}(r)$ with radius $r<R$. If $U$ is a compact, connected and simply
connected domain in the symplectic plane $(\mathbb{R}^{2},\sigma
)=(\mathbb{R}^{2},-\det)$ then $c_{\mathrm{GR}}(U,\sigma)$ is the area of $U$.
There exist infinitely many symplectic capacities (and $c_{\mathrm{GR}}$ is
the smallest of all) but they all agree on ellipsoids (see e.g.
\cite{HZ,physreps}). We are going to give an explicit formula below, but let
us first introduce the following notation and terminology. Let $M$ be a
positive definite $2n\times2n$ real matrix and consider the product $JM$. Its
eigenvalues are those of the antisymmetric matrix $M^{1/2}JM^{1/2}$ and are
thus of the type $\pm i\lambda_{\sigma,j}(M)$, $j=1,...,n$, with
$\lambda_{\sigma,j}(M)>0$. The numbers $\lambda_{\sigma,j}(M)$ are called the
$\sigma$- \emph{eigenvalues} of $M$; up to a simultaneous reordering of the
variables $x_{j}$ and $p_{j}$ one can always assume that $\lambda_{\sigma
,1}(M)\geq\cdot\cdot\cdot\geq\lambda_{\sigma,n}(M)$. The ordered set
\[
\operatorname*{Spec}\nolimits_{\sigma}(M)=(\lambda_{\sigma,1}(M),...,\lambda
_{\sigma,n}(M))
\]
is then called the $\sigma$- \emph{spectrum} of $M$ (when it is understood
that it is the standard symplectic structure which is used one speaks about
symplectic eigenvalues and symplectic spectrum). One proves the following
properties:%
\begin{equation}
M\leq M^{\prime}\Longrightarrow\lambda_{\sigma,j}(M)\leq\lambda_{\sigma
,j}(M^{\prime})\text{ , }j=1,...,n \label{mmprime}%
\end{equation}
where $M\leq M^{\prime}$ means that $M^{\prime}-M$ is semi-definite positive,
and
\begin{equation}
(\lambda_{\sigma,1}(M^{-1}),...,\lambda_{\sigma,n}(M^{-1}))=(\lambda
_{\sigma,n}(M)^{-1},...,\lambda_{\sigma,1}(M)^{-1}) \label{minv}%
\end{equation}
(see de Gosson \cite{Birk}, \S 8.3). Moreover there exists a symplectic matrix
$S\in\operatorname*{Sp}(2n,\sigma)$ such that
\begin{equation}
S^{T}MS=%
\begin{pmatrix}
\Lambda^{M} & 0\\
0 & \Lambda^{M}%
\end{pmatrix}
\label{sms}%
\end{equation}
where $\Lambda^{M}$ is the diagonal matrix $\operatorname{diag}(\lambda
_{\sigma,1}(M)\geq\cdot\cdot\cdot\geq\lambda_{\sigma,n}(M))$ (Williamson's
diagonal form, see \cite{Arnold,Birk,HZ,wi36}). Formula (\ref{sms}) implies
the following: let $\mathcal{M}$ be the phase space ellipsoid by inequality
$z^{T}Mz\leq1$. The the inverse image $S^{-1}(\mathcal{M})$ has the normal
form%
\begin{equation}
\sum_{j=1}^{n}\lambda_{\sigma,j}(M)(x_{j}^{2}+p_{j}^{2})\leq1. \label{normal}%
\end{equation}
One can of course replace the matrix $J$ above by the antisymmetric
non-degenerate matrix $\Omega=F^{T}JF$ ($F$ is defined as in Proposition
\ref{propf}). Considering as above the symplectic form $\omega(z,z^{\prime
})=-(z^{\prime})^{T}\Omega^{-1}z$ the $\omega$-spectrum of $M$ is the
decreasing sequence
\[
\operatorname*{Spec}\nolimits_{\omega}(M)=(\lambda_{\omega,1}(M),...,\lambda
_{\omega,n}(M))
\]
of positive numbers such that the $\pm i\lambda_{\omega,j}(M)$ ($\lambda
_{\omega,j}(M)>0$) are the eigenvalues of of $\Omega M$; the properties
(\ref{mmprime}) and (\ref{minv}) hold mutatis mutandis, replacing the
subscript $\sigma$ with $\omega$. We have:%
\begin{equation}
\operatorname*{Spec}\nolimits_{\omega}(M)=\operatorname*{Spec}%
\nolimits_{\sigma}(FMF^{T}). \label{specspec}%
\end{equation}
To prove this it is sufficient to show that $\Omega M$ and $J(FMF^{T})$ have
the same eigenvalues. Suppose $\Omega M=F^{T}JFMz=\lambda z$ for some complex
number $\lambda$ and $z\neq0$. This is equivalent to $JFMz=\lambda(F^{T}%
)^{-1}z$ and hence to $JFMF^{T}\left[  (F^{T})^{-1}z\right]  =\lambda\left[
(F^{T})^{-1}z\right]  $ which proves our claim.

\begin{proposition}
\label{prop3}Let $M$ be definite positive and consider the ellipsoids
$\mathcal{M}^{+}=\{z:z^{T}Mz\leq1\}$ and $\mathcal{M}^{-}=\{z:z^{T}M^{-1}%
z\leq1\}$ in $\mathbb{R}^{2n}$. For every symplectic capacity $c$ we have
\begin{equation}
c(\mathcal{M}^{+},\sigma)=\pi/\lambda_{\sigma,1}(M)\text{ \ , \ \ }%
c(\mathcal{M}^{-},\sigma)=4\pi\lambda_{\sigma,n}(M). \label{sc}%
\end{equation}

\end{proposition}

\begin{proof}
See for instance \cite{physreps,HZ} for a proof of the first formula
(\ref{sc}). The second formula follows from the first in view of the equality
(\ref{minv}).
\end{proof}

The first formula (\ref{sc}) shows that the symplectic capacity of
$\mathcal{M}$ is the area of the intersection of that ellipsoid with the
$x_{1},p_{1}$ plane once it has been put in normal form (\ref{normal}). This
again shows that the symplectic capacity is related to an area, and not to
volume. In fact, using the invariance of the action form one can restate the
first formula (\ref{sc}) in the following way:%
\begin{equation}
c(\mathcal{M},\sigma)=\oint\nolimits_{\gamma_{\min}}pdx=\tfrac{1}{2}%
\oint\nolimits_{\gamma_{\min}}pdx-xdp \label{action}%
\end{equation}
where $\gamma$ is the shortest Hamiltonian orbit carried by the surface of the
ellipsoid $\mathcal{M}$ (see de Gosson and Luef \cite{physreps}); the
symplectic capacity of an ellipsoid is thus explicitly expressed in terms of a
\emph{dynamical action}.

\section{Application to the Uncertainty Principle}

Let us now prove the main result of this article. Recall (formula
(\ref{cov1})) that the covariance matrix is the set
\begin{equation}
\mathcal{C}=\{z:(z-\bar{z})^{T}\Sigma^{-1}(z-\bar{z})\leq m_{0}^{2}\}
\label{cov2}%
\end{equation}
for some suitable choice of the real number $m_{0}$.

\begin{theorem}
The condition $\Sigma+i\Omega\geq0$ for the covariance matrix $\Sigma$ is
equivalent (for every symplectic capacity $c$) to the condition%
\begin{equation}
c(\mathcal{C},\omega)=c((F^{T})^{-1}\mathcal{C},\sigma)\geq\pi m_{0}%
^{2}\lambda_{n}(\Sigma) \label{covsymp}%
\end{equation}
where $\lambda_{n}(\Sigma)$ is the largest $\sigma$- eigenvalue of $\Sigma$.
If (\ref{covsymp}) is satisfied then we have
\begin{align}
\Delta x_{j}^{2}\Delta x_{k}^{2}  &  \geq\Delta(x_{j},x_{k})^{2}+a_{jk}^{2}\\
\Delta p_{j}^{2}\Delta p_{k}^{2}  &  \geq\Delta(p_{j},p_{k})^{2}+c_{jk}^{2}\\
\Delta x_{j}^{2}\Delta p_{k}^{2}  &  \geq\Delta(x_{j},p_{k})^{2}+b_{jk}^{2}.
\end{align}

\end{theorem}

\begin{proof}
Since translations are symplectomorphisms in any symplectic structure, the
ellipsoid $\mathcal{C}$ has the same symplectic capacity as the centered
ellipsoid $z^{T}\Sigma^{-1}z\leq m_{0}^{2}$. The result now follows from
Proposition \ref{prop3} above with $M=m_{0}^{-1}\Sigma^{-1}$.
\end{proof}

\section{Discussion}

The notion of symplectic capacity appears as a device allowing to measure in a
geometrical way the size of the MVE in a new way. The inequalities
\begin{align*}
\Delta x_{j}^{2}\Delta x_{k}^{2}  &  \geq\Delta(x_{j},x_{k})^{2}+a_{jk}^{2}\\
\Delta p_{j}^{2}\Delta p_{k}^{2}  &  \geq\Delta(p_{j},p_{k})^{2}+c_{jk}^{2}\\
\Delta x_{j}^{2}\Delta p_{k}^{2}  &  \geq\Delta(x_{j},p_{k})^{2}+b_{jk}^{2}.
\end{align*}
are not equivalent to the condition $c(\mathcal{C},\omega)\geq\pi m_{0}%
^{2}\lambda_{n}^{\Sigma}$ but are implied by it. Therefore, $c(\mathcal{C}%
,\omega)\geq\pi m_{0}^{2}\lambda_{n}(\Sigma)$ can be viewed as a stronger
--but natural-- version of the uncertainty principle. Its usefulness might
very well come from the fact that the condition $c(\mathcal{C},\omega)\geq\pi
m_{0}^{2}\lambda_{n}(\Sigma)$ is invariant under arbitrary symplectic
transformations (linear or not). It is in particular preserved under
Hamiltonian time evolution, since Hamiltonian flows consist of symplectomorphisms.

\begin{acknowledgement}
This work has been financed by the Austrian Research Agency FWF (Projektnummer P20442-N13).
\end{acknowledgement}

\end{document}